\documentclass[a4paper,UKenglish]{lipics-v2021}
\usepackage{booktabs}
\title{Undecidability of Reachability for Constraint Automata on Words}
\titlerunning{Undecidability of Reachability for Constraint Automata}
\author{Jakub Michaliszyn}{University of Wrocław}{}{}{}
\authorrunning{Jakub Michaliszyn}
\Copyright{Jakub Michaliszyn}
\ccsdesc[500]{Theory of computation~Formal languages and automata theory}
\ccsdesc[300]{Theory of computation~Logic and verification}
\keywords{constraint automata, reachability, prefix, suffix, infix, undecidability}
\hideLIPIcs

\newcommand{\pref}{\sqsubset_P}
\newcommand{\suff}{\sqsubset_S}
\newcommand{\incomp}{\perp_I}
\newcommand{\wpref}{\sqsubseteq_P}
\newcommand{\wsuff}{\sqsubseteq_S}
\newcommand{\winfix}{\sqsubseteq_I}
\newcommand{\word}[1]{\mathtt{#1}}
\newcommand{\step}{\xrightarrow{U}}

\begin{document}
\nolinenumbers
\maketitle

\begin{abstract}
We prove that reachability for constraint automata over finite binary
words is undecidable. Eight registers suffice, with conjunctive guards
using only equality, proper prefix, proper suffix, and infix
incomparability, without word constants or concatenation.
Finite-run nonemptiness, B\"uchi nonemptiness, and bounded B\"uchi
nonemptiness are undecidable as consequences.
For Boolean guards over equality, prefix, suffix, and prefix and
suffix incomparability, B\"uchi nonemptiness is PSPACE-complete
with arbitrarily many registers and word constants.
The upper bound also allows positive infix atoms and yields
PSPACE-completeness of prefix--suffix LTL satisfiability.
\end{abstract}
\newpage
\section{Introduction}
\label{sec:introduction}

Constraint automata and related constraint temporal logics have been
studied over numerical and ordered
domains~\cite{Cerans1994,Revesz2000,Revesz2002,DemriDSouza2007,DemriPeriodicity2006,DemriGascon2008,DemriGascon2009,Gascon2009,SegoufinTorunczyk2011},
tree
domains~\cite{CarapelleKartzowLohrey2016,CarapelleFengKartzowLohrey2017,DemriQuaasTrees2025,DemriGu2026},
and
words~\cite{KartzowWeidner2015,Weidner2016,DemriDeters2016,PetelerQuaas2022,DemriQuaasStrings2023,MichaliszynOtopWieczorek2023,MichaliszynOtopWieczorek2026},
with related synthesis
results~\cite{ExibardFiliotKhalimov2022,DauvierFiliotReynier2026}
and
surveys~\cite{DemriOverview2006,CarapelleLohrey2015,DemriQuaasSurvey2021}.
A recurring idea in decidability proofs is to replace concrete values
by symbolic comparison information and determine which symbolic runs
can be realised.

We establish two results over finite binary words.
Reachability is undecidable with eight registers and conjunctive guards
using equality, proper prefix, proper suffix, and infix incomparability,
without word constants. Finite-run nonemptiness and both ordinary
and bounded B\"uchi nonemptiness are therefore undecidable.
For Boolean prefix and suffix guards, including equality and
incomparability, B\"uchi nonemptiness is PSPACE-complete with
arbitrarily many registers and arbitrary word constants.
This upper bound also allows positive infix atoms and gives
PSPACE-completeness of prefix--suffix LTL satisfiability.

For binary words, the decidability result answers the question about
combining prefix and suffix posed by Demri and
Deters~\cite[Section~5]{DemriDeters2016} and reiterated by Peteler
and Quaas~\cite[Section~1]{PetelerQuaas2022}, including the extensions
to several variables and arbitrary word constants.
It also resolves the unrestricted B\"uchi nonemptiness question
for the infix-positive fragment left open by Michaliszyn, Otop,
and Wieczorek~\cite[Section~7]{MichaliszynOtopWieczorek2026}.
Our undecidability result answers negatively their question about
extending reachability to infix
incomparability~\cite[Section~4]{MichaliszynOtopWieczorek2026}.\footnote{The CONCUR~2023 version
\cite{MichaliszynOtopWieczorek2023} claimed PSPACE upper bounds
also in the presence of infix incomparability.
The later preprint
\cite[Section~4.1]{MichaliszynOtopWieczorek2026}
identifies an error in the constraint-joining argument
and restricts its positive results to infix-positive constraints.}

\paragraph*{Results in a nutshell.}
The undecidability proof starts from a familiar idea: represent counter
values by words that grow with each increment.
The key is to obtain a well-defined counter value
from simple comparisons between words.
We define a relation on pairs of words using only the allowed guards.
Although its steps may increase word lengths by different amounts,
all paths with the same start and endpoint have the same number of
steps, as shown in Lemma~2.
Thus the number of steps from a fixed anchor gives a counter value,
independently of how the pair was reached.
Every step increases this value by exactly one.
A decrement can then be certified by finding a reachable predecessor
through a search from the anchor.
Neither the representation of a counter value nor its predecessor
needs to be unique.

For decidability, we separate the prefix and suffix requirements
by representing each word with two words of equal length.
The main step is a single global reconstruction that turns suitable
pairs into words while preserving the required relations and constants.
The choice to preserve relations in only one direction is essential:
distinct pairs with the same first coordinate may have to become
properly prefix-comparable after reconstruction.
The other direction is not needed in the proof as we construct the appropriate pairs directly.

\section{Model and proof strategy}
\label{sec:model}

Fix the alphabet $\Sigma=\{0,1\}$. Each register stores a word in $\Sigma^*$,
possibly the empty word. A constraint automaton has finitely many control
states and finitely many transitions. For $n$ registers, a transition guard is
a conjunction of atoms $uRv$, where
\[
u,v\in\{x_1,\ldots,x_n,x'_1,\ldots,x'_n\}.
\]
Unprimed variables denote values before the transition and primed variables
denote values after it. Comparisons between two new values are allowed.
Initial register values are arbitrary.

We use the signature $\{=,\pref,\suff,\incomp\}$. The relations $u\pref v$ and
$u\suff v$ mean that $u$ is a proper prefix and a proper suffix of $v$,
respectively. The relation $u\incomp v$ means that neither word is an infix of
the other, where an infix is a contiguous substring. In particular,
$u\incomp u$ never holds. For example,
$\word{01}\incomp\word{00}$, whereas
$\word{01}\not\incomp\word{101}$.
We write $\wpref$, $\wsuff$, and $\winfix$ for the corresponding containment
relations with equality allowed. They are used only in the proof and do not
enlarge the transition signature.

A configuration consists of a control state and an $n$-tuple of words. A
transition can be taken when its guard is satisfied by the old and new tuples.
\emph{Reachability} asks whether a designated control state occurs in some
finite run starting at an initial state.

\begin{theorem}
\label{thm:main}
Reachability of a designated state in constraint automata over finite binary
words is undecidable. Eight registers suffice, with no word constants and with
transition guards that are conjunctions of atoms over
$\{=,\pref,\suff,\incomp\}$.
\end{theorem}

The proof has three parts. First, we define a relation $U$ on pairs of words and
prove that paths with the same start and end have the same length. We then
interpret counter values as path lengths from an arbitrary anchor and implement
counter operations. Finally, we reduce the halting problem for two-counter
machines to reachability. An explicit sequence of pairs supplies a run whenever
the machine halts.

\section{Pairs of words and the step relation}
\label{sec:step}

Let
\(
D=\{(a,b)\in(\Sigma^*)^2:a\incomp b\}.
\)
Both words of every pair in $D$ are nonempty, since the empty word is an infix
of every word. For $A=(a,b)\in D$ and $C=(c,d)\in D$, define
\begin{equation}
\label{eq:U}
U(A,C)\quad\Longleftrightarrow\quad
a\pref c\ \land\ b\suff c\ \land\ a\suff d\ \land\ b\incomp d.
\end{equation}
The first word of the new pair begins with $a$ and ends with $b$. The second
ends with $a$ and is infix-incomparable with $b$. Membership of the two pairs in
$D$ contributes $a\incomp b$ and $c\incomp d$. Thus the full expansion of $U$
is a conjunction of six allowed atoms.

For example,
\[
U\bigl((\word{01},\word{00}),(\word{0100},\word{0101})\bigr).
\]
Indeed, $\word{0100}$ begins with $\word{01}$ and ends with $\word{00}$,
while $\word{0101}$ ends with $\word{01}$ and contains no $\word{00}$.
Within each pair the words are distinct and have equal length, so they are
infix-incomparable.

A \emph{$U$-path} is a finite sequence $Q_0,\ldots,Q_r$ such that
$U(Q_i,Q_{i+1})$ holds at each step. Its length is $r$: we count steps, not
pairs. A path of length zero is allowed. Every step strictly lengthens the
first word.

\begin{lemma}
\label{lem:unique}
For any $Z,C\in D$, all $U$-paths from $Z$ to $C$ have the same length.
\end{lemma}

\begin{proof}
Consider two $U$-paths with the same initial and final pairs,
numbered by their distance from the endpoint:
\[
A=A_m\xrightarrow{U}\cdots\xrightarrow{U}A_0=C,
\qquad
A=B_n\xrightarrow{U}\cdots\xrightarrow{U}B_0=C.
\]
Write $A_i=(a_i,b_i)$ and $B_j=(c_j,d_j)$.
By symmetry, assume $m\le n$.
We compare $A_i$ with $B_i$: both positions have exactly
$i$ steps remaining before the common endpoint.

Working backwards along the paths, we prove by induction on $i$ that
\[
b_i\not\winfix c_i
\qquad \text{ for } i \in \{0, \dots, m\}.
\]
Here $\winfix$ allows equality.

For $i=0$, both paths are at the endpoint:
\(
(a_0,b_0)=(c_0,d_0)=C.
\)
Thus $c_0=a_0$. Since $C\in D$, its two coordinates are
infix-incomparable. In particular, $b_0\not\winfix a_0$,
and hence $b_0\not\winfix c_0$.

Now let $1\le i\le m$ and assume that
\(
b_{i-1}\not\winfix c_{i-1}.
\) 
The edge $B_i\xrightarrow{U}B_{i-1}$ gives
$c_i\pref c_{i-1}$. Any infix of $c_i$ is also an infix
of $c_{i-1}$. Hence \( b_{i-1}\not\winfix c_i. \)

{We show that $c_i$ is a proper suffix of $b_{i-1}$.}
Each edge $A_t\xrightarrow{U}A_{t-1}$ exchanges the
coordinates in its suffix conditions: 
\(
a_t\suff b_{t-1},
\qquad
b_t\suff a_{t-1},
\)
and likewise on the second path. 
Following these conditions,
a word remains a suffix of the current coordinate, but changes
from the first coordinate to the second, or from the second
to the first, at every step.

The word $b_{i-1}$ starts in the second coordinate with
$i-1$ steps remaining, whereas $c_i$ starts in the first
coordinate with $i$ steps remaining.
They therefore end in the same coordinate of $C$: $a_0$ ($=c_0$) if $i$ is even and $b_0$($=d_0$) if $i$ is odd. Suffixes of a fixed word are linearly ordered by weak
suffix containment, so we have
\( b_{i-1}  \wsuff c_i \text{ or } c_i \wsuff b_{i-1} \).
The relation $b_{i-1}\wsuff c_i$
 is impossible because it would contradict
$b_{i-1}\not\winfix c_i$. Consequently,
\(
c_i\suff b_{i-1}.
\)

We now use infix incomparability on the first path.
The edge $A_i\xrightarrow{U}A_{i-1}$ requires
\(
b_i\incomp b_{i-1}.
\) 
If $b_i\winfix c_i$, then $c_i\suff b_{i-1}$ would imply
$b_i\winfix b_{i-1}$, contradicting this requirement.
Thus
\(
b_i\not\winfix c_i,
\)
which completes the induction.

At $i=m$, we obtain $b_m\not\winfix c_m$.
Suppose that $n>m$.
Since $A_m=B_n$, we have $b_m=d_n$.
The first edge of the second path gives $d_n\suff c_{n-1}$,
and its prefix conditions give $c_{n-1}\wpref c_m$.
Therefore
\(
b_m=d_n\suff c_{n-1}\wpref c_m,
\) so $b_m\winfix c_m$, a contradiction.
Hence $m=n$.
\end{proof}

\section{Anchor and counter operations}
\label{sec:height}

Fix an arbitrary \emph{anchor} $Z\in D$. 
An arbitrary run of a constraint automaton can start with any words and in our logic without constants we cannot control it.  But once it starts with some pair of words, we will ensure that it belongs to $D$ and build on it. 

Let
\(
R_Z=\{P\in D:\text{there is a $U$-path from $Z$ to $P$}\}.
\)
For $P\in R_Z$, define its \emph{height} $h_Z(P)$ to be the length of any path
from $Z$ to $P$. Lemma~\ref{lem:unique} makes this independent of the chosen
path. Height is undefined outside $R_Z$.

The anchor is floor zero, and every $U$-step goes up one floor. Height counts
steps rather than letters. Directly from the definition,
$h_Z(P)=0$ if and only if $P=Z$
and if
$P\in R_Z \land U(P,P')$, then $P'\in R_Z\land h_Z(P')=h_Z(P)+1.$

We also need an explicit sequence providing every height for one possible
anchor. Define $P_r=(X_r,Y_r)$ by
\begin{equation}
\label{eq:witnesses}
X_0=\word{0},\quad Y_0=\word{1},\quad
X_{r+1}=X_rY_r,\quad Y_{r+1}=X_rX_r.
\end{equation}
The first pairs are 
$(\word{0}$, $\word{1})$,
$(\word{01}$, $\word{00})$,
$(\word{0100}$, $\word{0101})$, 
$(\word{01000101}$, $\word{01000100})$.

\begin{lemma}
\label{lem:witnesses}
For every $r\geq0$, we have $P_r\in D$, $U(P_r,P_{r+1})$, and
$h_{P_0}(P_r)=r$.
\end{lemma}

\begin{proof}
By induction, $X_r,Y_r$ have length $2^r$ and differ only in their last letter.
They are therefore distinct and infix-incomparable. The three prefix and suffix
conditions in $U(P_r,P_{r+1})$ follow from~\eqref{eq:witnesses}.

It remains to check $Y_r\incomp X_rX_r$. Every infix of length $|X_r|$ in
$X_rX_r$ is a cyclic shift of $X_r$ and hence has the same number of ones as
$X_r$. But $Y_r$ has a different number of ones, since it differs from $X_r$
in exactly the last letter. It cannot be such an infix. The reverse containment
is impossible by length. Thus $P_0\step\cdots\step P_r$ is a path of length $r$.
\end{proof}

The words $\word{0},\word{1}$ and the concatenations in~\eqref{eq:witnesses}
are used only to construct witnesses. They do not occur in transition labels:
the automaton checks only the atoms defining $U$ and equalities between
registers.

A height need not have a unique representation. For example, both
$(\word{0100},\word{0101})$ and $(\word{0100},\word{101})$ have height~2
relative to $P_0$. Both are reached through $(\word{01},\word{00})$. In the
second case, $\word{01}\suff\word{101}$ and
$\word{00}\incomp\word{101}$; the other conditions hold as well.
We do not need every pair of height $r$ to support all operations supported by
$P_r$. Correctness will hold for every completed operation, regardless of the
words chosen. To realise a halting computation, we will choose the pairs $P_r$.

\medskip

A \emph{counter} is a pair $P\in R_Z$, with value $h_Z(P)$. The anchor $Z$ is stored
in two registers and copied on every transition. Setting $P$ equal to $Z$
gives zero. The guard $U(P,P')$ increments the value by one, and $P=Z$ tests
for zero. Equality of pairs abbreviates two equality atoms.

Simply reversing $U$ need not preserve the domain of height. For example,
$P_2$ is reachable from $P_0$, and
\(
U\bigl((\word{01},\word{100}),P_2\bigr),
\)
yet $(\word{01},\word{100})\notin R_{P_0}$. Its first word has length~2, so
a path from $P_0$ could have only one step. That step is impossible, since
$\word{1}\not\incomp\word{100}$.

We therefore search for a predecessor by moving forwards
from the anchor. The search uses one auxiliary pair $S$,
shared by the two counters, and leaves $P$ unchanged until
a predecessor is found. The macro has three stages:
\begin{enumerate}
\item Enter the auxiliary state and initialise the search:
      \(
      S'=Z\ \land\ P'=P.
      \)
\item Take any number of forward steps, keeping the counter fixed:
      \(
      U(S,S')\ \land\ P'=P.
      \)
\item Exit when the current search pair is a predecessor of $P$:
      \(
      U(S,P)\ \land\ P'=S\ \land\ S'=Z.
      \)
\end{enumerate}
The assignments in stage~3 are simultaneous: the new counter
receives the old value of $S$, while the auxiliary pair is
reset to the anchor.

Every transition requires all pairs, before and after
the transition, to belong to $D$. The anchor and the other
counter are copied throughout. The auxiliary control state
has only the loop from stage~2 and the exit from stage~3.

To see that the macro implements decrement, let $P_{\mathrm{in}}$
be the input counter pair. It remains unchanged during the search.
If the macro completes after $t\ge0$ loop iterations, its successive
search values and the exit guard give a path
\[
Z=S_0\step S_1\step\cdots\step S_t\step P_{\mathrm{in}}.
\]
Thus $S_t,P_{\mathrm{in}}\in R_Z$, and uniqueness of path length gives $h_Z(S_t)=t$ and $h_Z(P_{\mathrm{in}})=t+1$.

The exit sets $P_{\mathrm{out}}=S_t$ and resets $S$ to $Z$.
Hence every completed execution decreases the counter height
by exactly one and has an input of positive height.

Conversely, if $P_{\mathrm{in}}\in R_Z$ has height $r>0$,
choose a path
\[
Z=Q_0\step Q_1\step\cdots\step Q_r=P_{\mathrm{in}}.
\]
Start with $S=Q_0$, take $r-1$ loop iterations through
$Q_1,\ldots,Q_{r-1}$, and exit using $U(Q_{r-1},P_{\mathrm{in}})$.
For $r=1$, no loop iteration is needed. Thus every reachable
positive-height input admits a completed decrement, whereas
a zero-height input admits none.

Throughout the search, $S$ belongs to $R_Z$. The search is
nondeterministic and need not terminate for every choice of forward
steps, but the counter is changed only at a certified exit.
No guard for membership in $R_Z$ or for positive height is required.

For example, with anchor $Z=P_1$, the pair $P_3$ has height~2 via
$P_1\step P_2\step P_3$. A decrement starts with $S=P_1$,
takes one forward step to $S=P_2$, and exits using $U(P_2,P_3)$.
The new counter is $P_2$, of height~1, and $S$ is reset to $P_1$.
The construction works for an arbitrary anchor; no particular words
are built into the guards.

\section{Reduction from two-counter machines}
\label{sec:reduction}

We use a deterministic machine with two counters in $\mathbb N$. Both counters start at zero. The
instructions are of the forms
\(
\operatorname{INC}(C_i);\ \mathrm{goto}\ q,
\)
and
\(
\text{if }C_i=0,\ \mathrm{goto}\ q_0;
\text{else }C_i:=C_i-1,\ \mathrm{goto}\ q_1.
\)
The halting problem for these machines is undecidable by Minsky's classical
universality result for two-register
machines~\cite[Chapter~14]{Minsky1967}. We direct
all halting instructions to one state $\mathsf{HALT}$.

Given a machine $M$, construct an automaton $\mathcal A_M$ with four pairs,
hence eight registers:
\[
Z=(z,w),\qquad A_1=(x_1,y_1),\qquad
A_2=(x_2,y_2),\qquad S=(s,t).
\]
On every transition, require all pairs before and after it to belong to $D$,
copy the anchor $Z$, and copy every unchanged pair. These requirements are
conjunctions of equality and infix-incomparability atoms.

The unique initial state $\mathsf{INIT}$ has a transition to the initial program
state setting
all initial pairs equal to the chosen anchor, representing the value 0: $
A'_1=Z, A'_2=Z, S'=Z.
$

Translate an increment of $C_i$ into a transition with guard $U(A_i,A'_i)$.
For a conditional instruction, use a zero branch with guard $A_i=Z$ leading
to $q_0$ and a decrement macro leading to $q_1$. Give each conditional
instruction its own auxiliary control state. Its four edges are the zero
branch, the entry into the macro, the loop, and the exit. The auxiliary pair
$S$ is shared. The construction is effective and has size linear in the size
of the machine.

\begin{proof}[Proof of Theorem~\ref{thm:main}]
Suppose first that $\mathcal A_M$ reaches $\mathsf{HALT}$. At every visit to a
program state, both counter pairs lie in $R_Z$. This follows inductively:
initialisation gives the anchor, an increment extends a $U$-path, and a
completed decrement macro supplies a path from the anchor to the new counter
pair. The other counter is copied.

We can therefore interpret the counters as $h_Z(A_1),h_Z(A_2)$. An increment
adds one, the zero branch requires exactly zero, and a decrement subtracts one
only from a positive value. A run reaching $\mathsf{HALT}$ must have left every
auxiliary state it entered, so all its macros completed. Removing those
auxiliary steps gives a valid machine computation reaching $\mathsf{HALT}$.
This direction holds for every anchor chosen by the run.

Conversely, suppose that $M$ halts. Choose $Z=P_0$ and represent
counter value $r$ by $P_r$. Lemma~\ref{lem:witnesses} realises
every increment, and equality with the anchor detects zero.
For a decrement from $P_r$ with $r>0$, initialise $S=P_0$,
take $r-1$ forward steps along the witness path to $P_{r-1}$,
and exit using $U(P_{r-1},P_r)$. This sets the counter to
$P_{r-1}$ and resets $S$ to $P_0$. For $r=1$, the exit is
taken immediately, without any loop iteration.
Thus every instruction of the halting computation can be
realised, and $\mathcal A_M$ reaches $\mathsf{HALT}$.
Only the existence of this choice of anchor is needed
in this direction.

Consequently,
$M$ halts if and only if 
$\mathsf{HALT}$ is reachable in $\mathcal A_M$.
Every transition label is a conjunction of allowed atoms, with no word constants
or concatenation. Undecidability of the halting problem proves the theorem.
\end{proof}

\section{Consequences}
\label{sec:consequences}

Both corollaries hold under the restrictions of Theorem~\ref{thm:main}: eight
registers, no word constants, and the same signature.

\begin{corollary}
\label{cor:finite}
Nonemptiness of constraint automata for finite runs is undecidable.
\end{corollary}

\begin{proof}
Make the designated state the only final state. A finite accepting run exists
exactly when that state is reachable.
\end{proof}

\begin{corollary}
\label{cor:buchi}
B\"uchi nonemptiness of constraint automata is undecidable. It remains
undecidable when restricted to runs whose word lengths are uniformly bounded.
\end{corollary}

B\"uchi nonemptiness asks whether an infinite run visits accepting states
infinitely often. In the bounded variant, there must also be a single number
$B$ bounding the lengths of all words in the run. The bound is not part of the
input.

\begin{proof}
In $\mathcal A_M$, make $\mathsf{HALT}$ the only accepting state and give it
just one outgoing transition: a self-loop copying all registers. Every accepting
run must reach $\mathsf{HALT}$. Conversely, every finite run to that state
extends to an accepting run by taking this loop forever. Word lengths are then
uniformly bounded: the prefix is finite and the tail is constant. Thus
\[
\begin{aligned}
M\text{ halts}
&\iff \mathcal A_M\text{ has an accepting B\"uchi run}\\
&\iff \mathcal A_M\text{ has a bounded accepting B\"uchi run}.
\end{aligned}
\]
This proves both claims.
\end{proof}
\section{PSPACE-completeness without infix incomparability}
\label{sec:decidability}

We now consider arbitrary numbers of registers and remove infix
incomparability. In addition to proper prefix and proper suffix, we allow
prefix and suffix incomparability: $u\perp_Pv$ means that neither word is a
prefix of the other, and $\perp_S$ is defined analogously. These relations
also exclude equality. Word constants in this section are written explicitly
as binary words.

\begin{theorem}
\label{thm:pspace}
B\"uchi nonemptiness over finite binary words is PSPACE-complete for constraint
automata with Boolean transition guards over
$\{=,\pref,\suff,\perp_P,\perp_S\}$ and arbitrary word constants.
\end{theorem}

PSPACE-hardness can be proved by reducing a deterministic Turing machine deciding a
PSPACE-complete language using one register per memory cell, as in \cite{MichaliszynOtopWieczorek2026}. We focus on proving the membership.

The upper bound concerns all infinite runs, with no bound on word lengths.
We represent each register by two words of equal length. One word records its
prefix constraints; the other records its suffix constraints in reverse.
The following lemma reconstructs actual register values from all these pairs
at once. We then reduce nonemptiness to the string temporal logic of
Demri and Deters~\cite{DemriDeters2016}.

\subsection{A global representation preserving constants}
\label{sec:global-representation}

\newcommand{\rev}{R}

Write $w^{\rev}$ for the reversal of $w$ and $X^{\rev} = \{x^{\rev} \mid x \in X\}$ for a set of words $X$. 
Let $C$ be the finite set of
data constants in the automaton and let $K$ contain $\epsilon$, $C$, and all the infixes of elements of $C$, i.e., \(
K=\{\varepsilon\}\cup
\{u:\exists c\in C\;\;u\sqsubseteq_I c\}.
\)

Define
\begin{equation}
\label{eq:coherent-domain}
E_K={}\{(k,k^{\rev}):k\in K\}
{}\cup\{(p,s):|p|=|s|,\ p\notin K,\ s\notin K^{\rev}\}.
\end{equation}
Equivalently, membership in $E_K$ means equal lengths together with
\begin{equation}
\label{eq:coherence}
\operatorname{Coh}_K(p,s)=
\bigwedge_{k\in K}(p=k\ \leftrightarrow\ s=k^{\rev}).
\end{equation}
Every pair $(w,w^{\rev})$ belongs to $E_K$.

\begin{lemma}
\label{lem:global-representation}
For every finite $K\subseteq\{0,1\}^*$ containing
$\varepsilon$ and closed under prefixes and suffixes,
there is an injection
$H_K:E_K\to\{0,1\}^*$
such that $H_K(k,k^{\rev})=k$ for $k\in K$.
For $a=(p,s),b=(p',s')\in E_K$, it satisfies
\begin{equation}
\label{eq:representation-properties}
\begin{aligned}
p\pref p'
&\ \Longrightarrow\ H_K(a)\pref H_K(b),
\\
s\pref s'
&\ \Longrightarrow\ H_K(a)\suff H_K(b),
\\
p\perp_P p'
&\ \Longrightarrow\ H_K(a)\perp_P H_K(b),
\\
s\perp_P s'
&\ \Longrightarrow\ H_K(a)\perp_S H_K(b).
\end{aligned}
\end{equation}
\end{lemma}

\begin{proof}
We construct $H_K(p,s)$ in stages, first for pairs with
coordinates of length~$0$, then length~$1$, and so on.
The idea is to satisfy the prefix requirements at the
left end of the constructed word and the suffix
requirements at its right end.

For this purpose, we keep two families of auxiliary words.
For each first coordinate $p$, the word $L_p$ has every
previously constructed value $H_K(p,s)$ as a prefix.
For each second coordinate $s$, the word $R_s$ has every
previously constructed value $H_K(p,s)$ as a suffix.
Thus a new word beginning with $L_p$ and ending with $R_s$
will preserve all requirements accumulated for these
coordinates. We obtain such a word simply by concatenation.

We first prescribe the values that must remain unchanged. For each $k\in K$ we set
\(
H_K(k,k^{\rev})=k.
\)

The auxiliary words for the empty coordinates are
\(
L_\varepsilon=R_\varepsilon=\varepsilon.
\)

Suppose that all pairs with coordinates shorter than $n$
have been processed. To start the length-$n$ stage, for all $u$ s.t. $|u|=n-1$ set
\[
L_{u0}=L_u0, \qquad
L_{u1}=L_u1, \qquad
R_{u0}=0R_u, \qquad
R_{u1}=1R_u.
\]

Now process, in lexicographic order, all pairs $(p,s)\in E_K$
with $|p|=|s|=n$ and $p\notin K$.
For each pair, define
\(
H_K(p,s)=L_pR_s
\)
and update the auxiliary words by
\(
L_p\gets H_K(p,s),
R_s\gets H_K(p,s).
\)

This update preserves their intended meaning:
the old $L_p$ is a prefix of $L_pR_s$, and the old
$R_s$ is a suffix of $L_pR_s$. Hence all earlier values
for these coordinates remain prefixes or suffixes,
respectively, of the new auxiliary words.
There are only finitely many pairs at each stage, so
every value of $H_K$ is a finite binary word.

We must also check compatibility with the prescribed
values. If $k\in K$, the definition of $E_K$ allows only
the pair $(k,k^{\rev})$ with first coordinate $k$
or second coordinate $k^{\rev}$. Consequently, no update
changes $L_k$ or $R_{k^{\rev}}$.
Induction on $|k|$ gives
\[
L_k=k=R_{k^{\rev}}.
\]
Indeed, the prefix of $k$ obtained by deleting its last
bit belongs to $K$, so the initialisation of $L_k$
reconstructs $k$. Likewise, the suffix obtained by
deleting the first bit belongs to $K$, so the
initialisation of $R_{k^{\rev}}$ reconstructs $k$.
Thus the prescribed value has exactly the required
prefix and suffix.

In the remainder of the proof, $L_p$ and $R_s$ refer
to their final values after their stage.
The construction gives for every $(p,s)\in E_K$:
\[
H_K(p,s)\wpref L_p,
\qquad
H_K(p,s)\wsuff R_s.
\]
It also gives the following separation property:
every value whose first coordinate has prefix $u\delta$
begins with $L_u\delta$, and every value whose second
coordinate has prefix $u\delta$ ends with $\delta R_u$.
For coordinates equal to $u\delta$, this follows from
initialisation and preservation under updates.
For longer coordinates, it follows by repeating the
same argument at subsequent stages.

We can now verify the four required implications.
Suppose first that $p\pref p'$, and let $\delta$ be
the bit of $p'$ immediately following $p$.
Then $H_K(p,s)$ is a prefix of $L_p$, while
$H_K(p',s')$ begins with $L_p\delta$. Therefore
\[
H_K(p,s)\pref H_K(p',s').
\]
Similarly, if $s\pref s'$ and $\delta$ is the next bit
after $s$, then $H_K(p,s)$ is a suffix of $R_s$,
while $H_K(p',s')$ ends with $\delta R_s$. Hence
\[
H_K(p,s)\suff H_K(p',s').
\]

Next, suppose that $p\perp_P p'$.
At the first position where the coordinates differ,
one has bit~$0$ and the other bit~$1$.
Writing $u$ for their longest common prefix, their
constructed values begin with $L_u0$ and $L_u1$,
in some order. They therefore differ immediately after
$L_u$, so neither is a prefix of the other.
If $s\perp_P s'$, the same argument applies from
the right: the values end with $0R_u$ and $1R_u$
for the longest common prefix $u$ of $s$ and $s'$,
and neither is a suffix of the other.
This proves~\eqref{eq:representation-properties}.

Finally, $H_K$ is injective.
For two pairs with different first coordinates,
proper prefix containment or prefix incomparability
forces their values to differ.
If the first coordinates are equal but the pairs
are different, their second coordinates are distinct
words of equal length. They are therefore
prefix-incomparable, which forces the constructed
values to be suffix-incomparable and hence distinct.
\end{proof}

The implications in~\eqref{eq:representation-properties}
are one-way: pairs with the same first coordinate may
acquire a proper prefix relation. This is sufficient
for positive atoms; negations will be eliminated before
translation.

The equal-length condition prevents conflicting prefix
and suffix requirements. Without it, the abstract pairs
$
u=(\word{0},\word{00})$,
$v=(\word{00},\word{0})$
would satisfy both $u\pref v$ and $v\suff u$, which
no actual words can satisfy.

Coherence is needed to preserve constants. Without it,
the unsatisfiable constraints
\[
u\pref\word{00}
\ \land\
v\pref\word{00}
\ \land\
u\perp_S v
\]
would admit the abstract pairs
$u=(\word{0},\word{0})$ and
$v=(\word{0},\word{1})$, with $\word{00}$ represented by
$(\word{00},\word{00})$.
Here $K=\{\varepsilon,\word{0},\word{00}\}$,
and~\eqref{eq:coherence} excludes the second pair.

\subsection{Prefix-length temporal logic}

We use the temporal logic
$\operatorname{LTL}(\{0,1\}^*,\operatorname{clen})$,
where $\operatorname{clen}(u,v)$ denotes the length of the
longest common prefix of the finite binary words $u$ and $v$.

A model is an infinite sequence
$\nu=(\nu_t)_{t\geq 0}$ of valuations assigning a finite binary
word to each variable. A word term is either a constant word
or an expression $\mathbf X^j x$ with $j\geq 0$.
At position $t$, a constant denotes itself and
$\mathbf X^j x$ denotes $\nu_{t+j}(x)$;
we identify $\mathbf X^0x$ with $x$.
Atomic formulae have the form
\[
\operatorname{clen}(u,v)
\leq
\operatorname{clen}(u',v'),
\]
where $u,v,u',v'$ are word terms.
Formulae are built from these atoms using Boolean connectives
and the temporal operators $\mathbf X$ (next) and $\mathbf U$
(until). At position $t$, the formula $\mathbf X\varphi$
holds if $\varphi$ holds at position $t+1$.
The formula $\varphi\,\mathbf U\,\psi$ holds if there is
a position $t'\geq t$ at which $\psi$ holds, and $\varphi$
holds at every position $s$ with $t\leq s<t'$.
We use the standard abbreviations
\(
\mathbf F\varphi=\top\,\mathbf U\,\varphi\) and 
\(
\mathbf G\varphi=\neg\mathbf F\neg\varphi.
\)
Thus $\mathbf F\varphi$ means that $\varphi$ holds eventually,
and $\mathbf G\varphi$ means that it holds at every current
and future position. 
We also write $\operatorname{clen}(u,v)
<
\operatorname{clen}(u',v')$ to denote $(\operatorname{clen}(u,v)
\leq
\operatorname{clen}(u',v')) \land \neg( \operatorname{clen}(u',v')
\leq
\operatorname{clen}(u,v))$.
Satisfiability asks whether a formula
holds at position~$0$ of some model.

We write \(\ell(u)=\operatorname{clen}(u,u)=|u|.\)
Observe that $u$ is a prefix of $v$
if and only if $\operatorname{clen}(u,v)=\ell(u)$.
We will also use the following constant-size abbreviations
for relations between word terms:
\begin{equation}
\label{eq:clen-abbreviations}
\begin{aligned}
u=v
&\ \Longleftrightarrow\
\ell(u)\leq\operatorname{clen}(u,v)
\ \land\
\ell(v)\leq\operatorname{clen}(u,v),
\\
u\pref v
&\ \Longleftrightarrow\
\ell(u)\leq\operatorname{clen}(u,v)
\ \land\
\ell(u)<\ell(v),
\\
u\perp_P v
&\ \Longleftrightarrow\
\operatorname{clen}(u,v)<\ell(u)
\ \land\
\operatorname{clen}(u,v)<\ell(v),
\\
|u|=|v|
&\ \Longleftrightarrow\
\ell(u)\leq\ell(v)
\ \land\
\ell(v)\leq\ell(u).
\end{aligned}
\end{equation}
Here $\pref$ denotes proper prefix and $\perp_P$ denotes
prefix incomparability.

Demri and Deters~\cite[Theorem~1]{DemriDeters2016}
prove that satisfiability of
$\operatorname{LTL}(\{0,1\}^*,\operatorname{clen})$
is PSPACE-complete.

\subsection{Reduction to prefix-length temporal logic}
\label{sec:pspace-reduction}
 We obtain our upper bound by a
polynomial-size reduction to the satisfiability problem of
$\operatorname{LTL}(\{0,1\}^*,\operatorname{clen})$.

Let $A=(Q,Q_0,F,\Delta)$ be the automaton, with non-empty $Q_0, F$, registers
$x_1,\ldots,x_n$ and set $C$ of data constants.
We construct a temporal formula $\Phi_A$ that is satisfiable
exactly when $A$ has an accepting infinite run.
We use the set $K$, the coherence condition
$\operatorname{Coh}_K$, and the domain $E_K$ defined above.

First normalise Boolean guards. For $R\in\{P,S\}$, exactly
one of equality, the two proper containment directions,
and incomparability holds. After pushing negations to atoms,
use
\begin{equation}
\label{eq:guard-normalisation}
\begin{aligned}
\neg(u\sqsubset_R v)
&\ \equiv\
u=v\ \lor\ v\sqsubset_R u\ \lor\ u\perp_R v,
\\
u\ne v
&\ \equiv\
u\pref v\ \lor\ v\pref u\ \lor\ u\perp_P v,
\\
\neg(u\perp_R v)
&\ \equiv\
u=v\ \lor\ u\sqsubset_R v\ \lor\ v\sqsubset_R u.
\end{aligned}
\end{equation}
This yields positive combinations of atoms with only a
constant-factor size increase. 
In what follows, every guard is assumed to be
normalised in this way.

For each register $x_i$, introduce two word variables
$p_i$ and $s_i$. The first coordinate is used to express
prefix relations, and the second to express suffix
relations as prefix relations on reversed words.

The current value of $x_i$ is represented by $(p_i,s_i)$.
Its value after the transition, denoted by $x_i'$, is
represented by the values of these variables at the next
position, namely $(\mathbf X p_i,\mathbf X s_i)$.
A constant word $c\in C$ is represented by the fixed pair
$(c,c^{\rev})$.

For any source term $u$, write $p_u$ and $s_u$ for the
first and second terms of its representing pair.
These are abbreviations: for example,
$p_{x_i'}=\mathbf X p_i$ and $s_c=c^{\rev}$.
Constants introduce no additional variables.

Translate positive atoms as follows,
and extend $T$ through $\land$ and $\lor$:

\begin{tabular}{lcl}
$u=v$
& $\to$
  & $p_u=p_v\ \land\ s_u=s_v$ \\
$u\pref v$
& $\to$
  & $p_u\pref p_v$ \\
$u\suff v$
& $\to$
  & $s_u\pref s_v$ \\
$u\perp_P v$
& $\to$
  & $p_u\perp_P p_v$ \\
$u\perp_S v$
& $\to$
  & $s_u\perp_P s_v$ 
\end{tabular}

The domain condition is
\[
\Psi_K=
\bigwedge_{i=1}^n
\bigl(
|p_i|=|s_i|
\ \land\
\operatorname{Coh}_K(p_i,s_i)
\bigr).
\]
Recall that $\operatorname{Coh}_K(p,s)=
\bigwedge_{k\in K}(p=k\ \leftrightarrow\ s=k^{\rev})$ and observe that since $K$ is of polynomial size, it can be enforced by enumerating all cases.

After expanding the definition of $\operatorname{Coh}_K$,
all relations in the translation and the domain condition
are expressed using~\eqref{eq:clen-abbreviations}.

Choose distinct binary constants $C_q$, for
$q\in Q$, of length logarithmic in $|Q|$.
Use one fresh word variable $\kappa$ for control, and write
$B_q$ for $\kappa=C_q$, with equality expanded
using~\eqref{eq:clen-abbreviations}.
Define $\Phi_A$ as the conjunction of the following formulas
\begin{itemize}
\item $\bigvee_{q\in Q_0}B_q$
\item $\mathsf G\left(
\Psi_K\ \land\!
\bigvee_{(q,\gamma,q')\in\Delta}
(B_q\land T(\gamma)\land\mathsf XB_{q'})\right)$
\item 
$\mathsf G\mathsf F\left(\bigvee_{q\in F}B_q\right)$
\end{itemize}

Empty disjunctions are interpreted as $\bot$.
The three conjuncts impose an initial state, a valid
transition at every position, and infinitely many visits
to accepting states, respectively.
The state codes occur only in the target logic: they are
not added to the data constants used to define $K$ and
are not reconstructed. The transition disjunction forces
a state code at every position, and distinct codes make
that state unique.

\begin{proof}[Proof of Theorem~\ref{thm:pspace}]
For an accepting run $(q_t,\bar w_t)_{t\geq 0}$, define
the valuations of a temporal model by
\[
\nu_t(p_i)=w_{t,i},
\qquad
\nu_t(s_i)=w_{t,i}^{\rev},
\qquad
\nu_t(\kappa)=C_{q_t}.
\]
Reversal changes suffix into prefix and preserves properness
and incomparability. Coherence holds automatically.
Thus these valuations satisfy $\Phi_A$.

Conversely, let $\nu=(\nu_t)_{t\geq 0}$ be a model of
$\Phi_A$. Read the unique control state $q_t$ at each
position and choose a transition witnessing the
corresponding true disjunct.
The global domain condition puts every register pair
in $E_K$. Using the single function from
Lemma~\ref{lem:global-representation}, 
for every $t\geq 0,\ 1\leq i\leq n$
set
\[
w_{t,i}
=
H_K\bigl(\nu_t(p_i),\nu_t(s_i)\bigr).
\]
Every true translated positive atom gives a true source
atom, including comparisons with constants.
Induction through $\land$ and $\lor$ preserves the
normalised guards. This applies equally to old values,
new values, and comparisons between them.
The control states and B\"uchi visits are unchanged.
Every reconstructed value is finite, and repeated pairs
always give the same word. We have proved
that 
$A$ has an accepting infinite run if and only if $\Phi_A\text{ is satisfiable}$.

The construction has polynomial size. The set $K$ contains
only the empty word and substrings of the data constants,
so both $K$ and the coherence conditions have polynomial
size. Guard translation increases size by only a constant
factor, and encoding the finite control also adds only polynomially many symbols. Unused registers
can be discarded, so binary register indices cause
no difficulty.

The PSPACE upper bound follows from the Demri--Deters
theorem. The words $H_K(p,s)$ are used only to prove
correctness and are never constructed by the algorithm.
\end{proof}
\subsection{Further consequences}
\label{sec:decidability-consequences}

\begin{corollary}
\label{cor:pspace-reachability}
Reachability and finite-run nonemptiness are PSPACE-complete under the
assumptions of Theorem~\ref{thm:pspace}.
\end{corollary}

\begin{proof}
For the upper bound, make each designated target or final state accepting and
add an identity self-loop there. B\"uchi nonemptiness is then equivalent to
reaching such a state. The Turing-machine construction in the proof of
Theorem~\ref{thm:pspace} gives hardness for both problems.
\end{proof}

\begin{corollary}
\label{cor:positive-infix}
Theorem~\ref{thm:pspace} and Corollary~\ref{cor:pspace-reachability} remain
valid when proper infix atoms are also allowed, provided they occur only
positively. Infix incomparability is excluded.
\end{corollary}

\begin{proof}
For each positive infix atom $u\sqsubset_I v$, introduce
a fresh register $z$ and replace the atom by
\(
u\sqsubset_P v
\ \lor\
u\sqsubset_S v
\ \lor\
\bigl(u\sqsubset_S z\land z\sqsubset_P v\bigr).
\)
Indeed, if $v=aub$, an occurrence at either boundary
satisfies one of the first two disjuncts. 
Use a separate fresh register for each occurrence,
with no copying requirements. Its value can therefore
be chosen independently at every transition.
Since infix atoms occur only positively, this replacement
preserves reachability, finite-run nonemptiness, and
B\"uchi nonemptiness.
\end{proof}

Thus the upper bound also covers full B\"uchi nonemptiness for the
infix-positive fragment considered in
\cite{MichaliszynOtopWieczorek2026}. The reconstruction may create extra
infixes; it does not preserve infix incomparability or negated infix atoms.

\begin{corollary}
\label{cor:prefix-suffix-ltl}
Satisfiability of LTL over finite binary words with prefix, suffix, equality,
and arbitrary explicitly written word constants is PSPACE-complete.
\end{corollary}

\begin{proof}
Put the formula in negation normal form, using the release operator if
necessary, and replace negative data atoms
by~\eqref{eq:guard-normalisation}. Apply $T$ to the atoms and conjoin
$\mathsf G\Psi_K$; no state variable is needed. Forward translation uses
$(w,w^{\rev})$, and reverse translation uses $H_K$ at all positions.
The implication for atoms extends by induction through the positive Boolean
and temporal operators $\mathsf X$, $\mathsf U$, and release. The reduction
has polynomial size, so the same Demri--Deters theorem supplies the upper
bound. PSPACE-hardness already holds for the prefix fragment by that theorem.
\end{proof}

\bibliography{references}
\appendix
\section{Statement on the use of generative AI}
This work was developed through close collaboration between the author
and several generative AI models, primarily GPT-6.1 Sol,
GPT-6 Astra Pro, and GPT-6 Astra. Generative AI made substantive contributions
to both main proofs.

In the undecidability proof, generative AI proposed the definition
of the relation $U$ together with a proof substantially more complicated
than the one presented here. The author considerably simplified
the argument and revised its presentation to make it accessible
to human readers. Similarly, in the decidability proof, generative AI
proposed the key definition of the reconstruction function $H_K$ (the proof proposed by AI did not include constants, so it was just $H$)
and an initially very complicated proof. The author progressively
simplified this proof into its present form.

Generative AI was also used to assist with mathematical and linguistic
audits, language refinement, and editorial tasks. The author takes
full responsibility for the mathematical claims and the final text.
\end{document}